\documentclass[12pt,reqno]{amsart}

\usepackage{array}
\usepackage{graphics}
\usepackage{amssymb, amsthm, amsmath, amsfonts,amscd}

\usepackage{multicol}
\usepackage{xfrac}
\usepackage{enumitem}
\usepackage{fancyhdr}
\usepackage{hyperref}
\usepackage{comment}
\usepackage{tikz-cd}
\usetikzlibrary{calc, positioning, matrix, arrows, decorations.pathmorphing, shapes, backgrounds}
\usepackage{pgf}
\usepackage{pgfplots}
\pgfplotsset{compat=1.6}
\usepackage{lpic}
\usepackage[numbers,sort&compress]{natbib}
\usepackage{pgfornament}

\theoremstyle{plain} 
\newtheorem{theorem}{Theorem}[section]

\newtheorem{lemma}[theorem]{Lemma}
\newtheorem{cor}[theorem]{Corollary}

\theoremstyle{remark}
\newtheorem{exr}{Exercise}

\theoremstyle{definition}

\newcommand{\R}{\mathbb R}
\renewcommand{\P}{\mathbb P}

\newcommand{\W}{\mathbf W}

\newcommand{\Y}{\mathbf{Y}_{\text{in}}}
\DeclareMathOperator*{\argmax}{arg\,max}

\newcounter{sarrow}

\newsymbol\dnd 232D

\newcommand{\ignore}[1]{}

\newcommand{\squishlist}
{
\begin{list}{$\bullet$}
{
\setlength{\parskip}{0pt} \setlength{\topsep}{0pt} \setlength{\partopsep}{0pt}
\setlength{\parsep}{0pt}  \setlength{\itemsep}{0pt}
}
}
\newcommand{\squishend}{\end{list}}

\begin{document}

\pagestyle{plain}{
\lhead{} \chead{} \rhead{}
\lfoot{} \cfoot{\thepage} \rfoot{}
}


\title{Noise Quantification and Control in Circuits via Strong Data-Processing Inequalities}
\author{Chenyang Sun}
\date{\today}

\begin{abstract}
This essay explores strong data-processing inequalities (SPDI's) as they appear in the work of Evans and Schulman \cite{ES} and von Neumann \cite{vN} on computing with noisy circuits. We first develop the framework in \cite{ES}, which leads to lower bounds on depth and upper bounds on noise that permit reliable computation. We then introduce the $3$-majority gate, introduced by \cite{vN} for the purpose of controlling noise, and obtain an upper bound on noise necessary for its function. We end by generalizing von Neumann's analysis to majority gates of any order, proving an analogous noise threshold and giving a sufficient upper bound for order given a desired level of reliability. 

The presentation of material has been modified in a way deemed more natural by the author, occasionally leading to simplifications of existing proofs. Furthermore, many computations omitted from the original works have been worked out, and some new commentary added. The intended audience has a rudimentary understanding of information theory similar to that of the author.
\end{abstract}
\maketitle
\thispagestyle{empty}
\numberwithin{equation}{section}


\section{Introduction to SPDI's and Noisy computing} \label{sect:1}

In information theory, data processing inequalities (DPI's) describe the loss of information under transformation. In a way, information behaves like some physical quantity that satisfies a one-sided conservation law, stipulating that it can be destroyed but not created.

DPI's take various forms as different applications call for information to be measured in different ways; some common measures include Shannon entropy (amount of uncertainty resolved in learning the outcome of a random variable, or message source), Kullback-Leibler divergence (the distinguishability of one distribution from another), and mutual information (the amount of dependence between the input and output of a communication channel). In this essay we explore mutual information in the context of analyzing circuits built from \textit{noisy}, or \textit{error}-prone, components. We use the term \textit{erroneous} to describe the output of a component, and \textit{noisy} to describe the component itself, where \textit{noise} is often quantified by the probability of \textit{error}.

A basic form of the DPI for mutual information is a standard result and can be found as Theorem 2.8.1 in \cite{TC}:
\begin{theorem}\label{thm: dpi}
    Let $X, Y, Z$ be random variables taking finitely many values, with $X\--Y\--Z$ forming a Markov chain. It holds that
    \begin{align}
        I(X;Z)\le I(X;Y).
    \end{align}
\end{theorem}

In the context of circuit design, the Markov chain $X\--Y\--Z$ models the operation of a noisy gate that takes $X$ as an input, performs computation that results in $Y$, and outputs $Z$, a version of $Y$ corrupted by some noise that is directly applied to $Y$, hence the Markov chain condition. In practice, the noise is often modeled by a conditional distribution $p_{Z\mid Y}$.

Theorem \ref{thm: dpi} can be interpreted qualitatively as saying that the noised output $Z$, compared to the intended output $Y$, is less dependent on $X$. This unsurprising interpretation captures the expectation that the output of a computation should depend on its input, and that the presence of noise weakens this dependence. 

In practice, however, qualitative inequalities fail to quantify how much information is lost, which is necessary for estimating the reliability of a circuit and designing noise-controlling strategies. Their quantified versions are known as strong data-processing inequalities (SPDI's), which often take the general form:
\[\frac{\text{information post-processing}}{\text{information pre-processing}}\le c,\]
where $c\in[0,1)$ is an upper bound on the proportion of information transmitted by the processing. This upper bound is usually made universal, in the sense that it depends solely on the properties of the processing channel, and not on the distributions that pass through.

These quantified inequalities naturally arise from many settings, from information percolation on directed graphs to various problems in constrained estimation; Chapter 33 in \cite{PW} gives a brief overview of various applications. This essay explores SPDI's that arise from the decay of information signals transmitted through noisy media, which is a suitable description of computation done by circuits built from noisy components.

Von Neumann introduced the ``noisy circuit" model of computation in an effort to study the limits of physical circuits. Since large-scale computation relies essentially on the propagation of long chains of events \cite{vN}, it is a legitimate concern that noise might accumulate to such an extent that the output of computation has little dependence on the input. Thus, it is of interest to construct circuits that can prevent the noise of unreliable components from accumulating to the point of catastrophic failure, and carry out computation with high probability of correctness. To this end, it is natural to ask:
\begin{itemize}
    \item Is such a construction possible?
    \item If so, at what cost?
\end{itemize}
The first part of this essay, based on \cite{ES}, addresses the latter question by showing that if such an noise-controlling construction is possible, the circuit is necessarily deeper (and thus less spatially and temporally efficient) compared to an noise-free circuit. 

The second part of this essay, based on \cite{vN}, demonstrates that under certain conditions, the construction of an noise-controlling circuit is possible in principle, albeit costly.

\section{Noise Quantification in Circuits}

\subsection{SPDI for gates as binary channels}

To model and analyze a circuit as a whole, it is helpful to first examine the processes that take place in individual gates. In this model, a gate takes a collection of binary inputs $X$, computes a single binary result $Y$, and outputs some noisy binary result $Z$. For now we restrict $X$ to a single binary input; the more general case readily follows.

To quantify the loss of information going from $Y$ to $Z$, one needs a model for the noise. Since the noise is applied directly to $Y$, it can be modeled as a conditional distribution $p_{Z\mid Y}$. Since $Y$ and $Z$ are both binary random variables, the behavior of noise can be completely described by two ways in which an error can happen: 
\begin{align}
    a:=&\ \P(Z=1\mid Y=0)=p_{Z\mid Y}(1\mid 0),\\
    b:=&\ \P(Z=0\mid Y=1)=p_{Z\mid Y}(0\mid 1),
\end{align}
where $a$ (resp. $b$) is the probability that a $0$ (resp. $1$) from $Y$ is corrupted.

In this essay, $H(B)$ shall denote the entropy of a discrete random variable $B$, and $h(p)$ the binary entropy function at $p$. Thus for a Bernoulli distribution $B_p$ with parameter $p$, it holds that $H(B_p)=h(p)$.

To prove a stronger, quantified version of the data processing inequality (Theorem \ref{thm: dpi}), it is helpful to examine the quantities $I(X;Y)$ and $I(X;Z)$ in this setting. For $X$ a Bernoulli random variable, one can rewrite
\begin{align}
    I(X;Y)=&\ H(Y)-H(Y\mid X)\\
    =&\ H(Y)-p_X(0)H(Y\mid X=0)-p_X(1)H(Y\mid X=1).
\end{align}
Taking $p:=p_X(0)$, $y_0:=p_{Y\mid X}(0\mid 0)$, and $y_1:=p_{Y\mid X}(0\mid 1)$, the above can be re-written as
\begin{align}
    &\ h(p_Y(0))-ph(y_0)-(1-p)h(y_1)\\
     =&\ h(py_0+(1-p)y_1)-ph(y_0)-(1-p)h(y_1).
\end{align}

By performing the above computation with $Z$ in place of $Y$, one has
\begin{align}
    I(X;Z)=&\ h(p_Z(0))-ph(z_0)-(1-p)h(z_1).
\end{align}
where $z_i:=p_{Z\mid X}(0\mid i)$. 

One can express $z_i$ using $y_i$ via expanding through $Y$ using the law of total expectation, and then applying the Markov chain property:
\begin{align}
    &\ \P(Z=0\mid X=i)\\
    =&\ \P(Z=0\mid X=i,Y=0)\P(Y=0\mid X=i)+\P(Z=0\mid X=i,Y=1)\P(Y=1\mid X=i)\\
    =&\ \P(Z=0\mid Y=0)\P(Y=0\mid X=i)+\P(Z=0\mid Y=1)\P(Y=1\mid X=i)\\
    =&\ (1-a)y_i+b(1-y_i).
\end{align}

For notational convenience, define $A(t):=(1-a)t+b(1-t)$, so that $z_i=A(y_i)$. Observe that by dropping the conditioning on $X$, the above calculations also yield $p_Z(0)=A(p_Y(0))$. We thus have
\begin{align}
    I(X;Z)=&\ (h\circ A)(p_Y(0))-p(h\circ A)(y_0)-(1-p)(h\circ A)(y_1)\\
    =&\ (h\circ A)(py_0+(1-p)y_1)-p(h\circ A)(y_0)-(1-p)(h\circ A)(y_1)
\end{align}

One may observe that this expression bears much resemblance to the definition of convexity. To take advantage of this observation, for a function $f$, arguments $x,y$ in the domain of $f$, and $t\in[0,1]$, we define a convexity evaluation
\begin{align}
    \Delta^2_{f}(x,y;t):=f(tx+(1-t)y)-tf(x)-(1-t)f(y),
\end{align}
so that we may write
\begin{align}\label{eqn: mutual}
    \frac{I(X;Z)}{I(X;Y)}=\frac{\Delta^2_{h\circ A}(y_0,y_1;p)}{\Delta^2_h(y_0,y_1;p)}.
\end{align}

We wish to obtain an upper bound for the ratio of convexity evaluations, over the parameters $y_0, y_1, p$; this can be given by a ratio of second derivatives through a general variational lemma, which is a version of Lemma 1 from \cite{ES}.
\begin{lemma}\label{lem: ratiosup}
    Let $f,g:[0,1]\to\R$ be twice-differentiable and convex. Then it holds that
    \begin{align}
        \sup_{\substack{x\ne y\in[0,1] \\ t\in(0,1)}}\frac{\Delta^2_f(x,y;t)}{\Delta^2_g(x,y;t)}=\sup_{t\in(0,1)}\frac{f''(t)}{g''(t)}.
    \end{align}
    Equality is attained as $x,y$ approach a value of $t$ achieving the right-hand supremum.
\end{lemma}
\begin{proof}
    See \cite{ES}, Lemma 1. 
\end{proof}

Applying Lemma \ref{lem: ratiosup} with $f=h\circ A$, $g=h$, $(x,y,t)=(y_0,y_1,p)$, one obtains
\begin{align}\label{eqn: 2deriv}
    \frac{I(X;Z)}{I(X;Y)}\le\sup_{p\in(0,1)}\frac{h''(A(p))(A'(p))^2}{h''(p)},
\end{align}
and thus we obtain Theorem 1 from \cite{ES}.
\begin{theorem}\label{thm: gatespdi}
    Let $X,Y$ be binary random variables, and $Z$ be the output of $Y$ through a noisy channel parametrized by error probabilities $a=p_{Z\mid Y}(1\mid 0)$, $b=p_{Z\mid Y}(0\mid 1)$. Then it holds that 
    \begin{align}
        \frac{I(X;Z)}{I(X;Y)}\le c_{a,b}
    \end{align}
    where 
    \begin{align}
        c_{a,b}=1-\left(\sqrt{b(1-a)}+\sqrt{a(1-b)}\right)^2.
    \end{align}
\end{theorem}
\begin{proof}
    The maximization of the right-hand side of \eqref{eqn: 2deriv} can be done explicitly, with 
    \begin{align}\label{eqn: optimalrate}
        p^*=\frac{\sqrt{b(1-b)}}{\sqrt{a(1-a)}+\sqrt{b(1-b)}}
    \end{align}
    yielding the claimed bound $c_{a,b}$.
\end{proof}

Although not discussed in the original paper, one might wonder whether an analogous lower bound exists; a discussion is postponed to the end of this subsection. For now, we extend Theorem \ref{thm: gatespdi} to take arbitrary inputs to obtain Corollary 1 from \cite{ES}.
\begin{cor}\label{cor: conditioned}
    Let $X,Y,Z$ be as in Theorem \ref{thm: gatespdi}, and $Q$ a random variable, not necessarily binary, such that
    \begin{align}
        (Q,X)\--Y\--Z
    \end{align}
    forms a Markov chain. Then
    \begin{align}
        \frac{I(X;Z\mid Q)}{I(X;Y\mid Q)}\le c_{a,b}.
    \end{align}
\end{cor}
\begin{proof}
    We first show that for any $q$, the distributions of $X,Y,Z$ given $Q=q$ satisfy the assumptions of Theorem \ref{thm: gatespdi}, and finish with an averaging argument.

    Since $Z$ is independent from $(Q,X)$ given $Y$, $Z$ is independent from each of $Q,X$ given $Y$. Thus
    \begin{align}
        p_{Z\mid X,Y,Q=q}=p_{Z\mid Y}=p_{Z\mid Y,Q=q},
    \end{align}
    so the conditional distributions given $Q=q$ still form a Markov chain. Furthermore, the second equality implies the error probabilities $p_{Z\mid Y}(1\mid 0)$, $p_{Z\mid Y}(0\mid 1)$ are unchanged after conditioning by $Q=q$. The conditional distributions of $X,Y,Z$ given $Q=q$ thus satisfy the assumptions of Theorem \ref{thm: gatespdi}, and it thus holds that
    \begin{align}
        \frac{I(X;Z\mid Q=q)}{I(X;Y\mid Q=q)}\le c_{a,b}.
    \end{align}
    Interpreting conditioning by $Q$ as taking a weighted average over $q$, we have
    \begin{align}
        \frac{I(X;Z\mid Q)}{I(X;Y\mid Q)}=&\ \frac{\sum_q p_Q(q)I(X;Z\mid Q=q)}{\sum_q p_Q(q)I(X;Y\mid Q=q)}\\
        \le&\ \max_q \frac{p_Q(q)I(X;Z\mid Q=q)}{p_Q(q)I(X;Y\mid Q=q)}\\
        \le&\ c_{a,b}.
    \end{align}
\end{proof}

Figure \ref{fig: nonseconst} plots $c_{a,b}$. Observe that $c_{a,b}$ attains a maximum of $1$ at $(0,0)$ and $(1,1)$, which corresponds to the channel being noiseless or perfectly erroneous. In either case the output $Z$ can be determined with certainty with knowledge of input $Y$ (and vice versa), hence the ineffective bound of $1$ for information transmission.

Furthermore, observe that the transmission factor $c_{a,b}$ attains a minimum of $0$ along $a+b=1$. Indeed: along this line we have
\begin{align}
    p_{Z\mid Y}(1\mid 0)=a=1-b= p_{Z\mid Y}(1\mid 1),
\end{align}
and similarly $p_{Z\mid Y}(0\mid 0)=p_{Z\mid Y}(0\mid 1)$, meaning that the input $Y$ does not affect the distribution of the output $Z$, i.e., $Z$ is independent of $Y$.

\begin{figure}
    \centering
    \includegraphics[width=12cm]{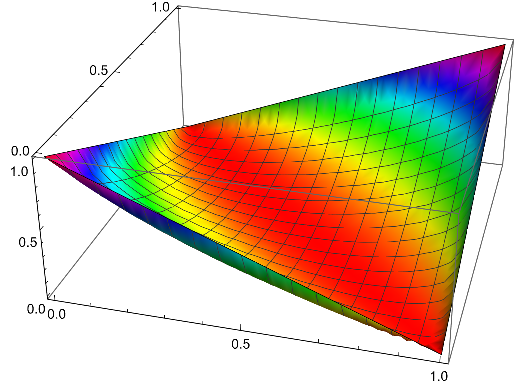}
    \caption{Plot of $c_{a,b}$ over $[0,1]\times[0,1]$, colored by height.}
    \label{fig: nonseconst}
\end{figure}

There are some interesting observations to be made regarding Lemma \ref{lem: ratiosup} and Theorem \ref{thm: gatespdi}. From Lemma \ref{lem: ratiosup}, one notices that the optimal transmission factor occurs when \begin{align}
    y_0=p_{Y\mid X}(0\mid 0),\ y_1=p_{Y\mid X}(0\mid 1)
\end{align}
are close to each other, and in particular to the quantity $p^*$ from Theorem \ref{thm: gatespdi}. The former observation, as Evans and Schulman remark, means that a signal must be weak in order to be transmitted at close to the optimal efficiency; however, weakness, i.e., closeness of $y_0,y_1$, does not imply close-to-optimal transmission. A counterexample is not given in \cite{ES}, but will be provided in the following discussion concerning a lower bound on the transmission factor.

One may also notice that the distribution of $Y$ resulting from optimal transmission ($y_0\approx y_1\approx p^*$), which is described by
\begin{align}
    \P(Y=0)=&\ p_{Y\mid X}(0\mid 0)p_X(0)+p_{Y\mid X}(0\mid 1)p_X(1)\\
    =&\ y_0p^*+y_1(1-p^*)\\
    \approx&\ p^*p^*+p^*(1-p^*)\\
    =&\ p^*=\frac{\sqrt{b(1-b)}}{\sqrt{a(1-a)}+\sqrt{b(1-b)}},
\end{align}
disagrees with the limiting distribution fixed by the channel, which satisfies $\P(Y=0)=a/(a+b)$ (see Lemma \ref{lem: limitingdist}). This may appear counterintuitive at first, but less so once one realizes that preserving a distribution is not equivalent to preserving information: two random variables can be identically distributed but independent and thus carry no mutual information.

One might naturally wonder whether a lower bound on the transmission factor, analogous to Theorem \ref{thm: gatespdi}, might be possible. It turns out that there is no nontrivial lower bound: as long as $a$ or $b$ is nonzero, there exist signals that are transmitted with arbitrarily low efficiency.

One can explicitly construct such signals. Suppose without loss of generality that $0<a<1$; as part of proving Lemma \ref{lem: ratiosup}, one verifies that for some fixed $t$,
\begin{align}
    \lim_{y\to t}\frac{\Delta^2_{h\circ A}(t,y;p)}{\Delta^2_h(t,y;p)}=\frac{(h\circ A)''(t)}{h''(t)}=\frac{h''(A(t))(1-a-b)^2}{h''(t)}.
\end{align}
Consider what happens for choices of $y_0, y_1$ very close to $1$, i.e., $Y$ sets $X$ to $0$ with high probability. Since $y_0$ and $y_1$ are close, by the above, the transmission factor $I(X;Z)/I(X;Y)$ is approximately
\begin{align}
    \frac{h''(A(y_0))(1-a-b)^2}{h''(y_0)}.
\end{align}
Taking $y_0\to 1$, we have $A(y_0)=(1-a)y_0+b(1-y_0)\to 1-a$, so the numerator approaches some constant (since $0<a<1$), while the denominator diverges to $-\infty$ and brings the information transmission factor arbitrarily close to $0$. This also exemplifies a weak signal that transmits with very low efficiency.

This conclusion is not surprising at a qualitative level: since $Y$ sets $X$ to $0$ with high probability, receiving a $0$ from $Y$ offers little information about $X$. The outcome of $X$ can be inferred with confidence only if the improbable event $Y=1$ occurs, which is readily overwhelmed by the false-positive $a=\P(Z=1\mid Y=0)$ when one can only observe $Z$.

\subsection{Minimum circuit depth and maximum component noise}

With sufficient understanding of what happens in each gate of a circuit, one can now analyze the behavior of the circuit as a whole. 

In this model of computation, a circuit consists of inputs and gates (collectively seen as vertices of a directed, acyclic graph) and wires carrying a binary value away from a vertex and into another, seen as edges directed from the outputting vertex to the accepting vertex. For clarity, wires may be identified with the information they carry, which is the output of its originating gate.

At this point, the original paper \cite{ES} makes the assumption that the noise channel is symmetric and remarks at the very end that the proofs still hold when the assumption is relaxed; in this essay we do not assume symmetry altogether.

It is helpful to first quantify the flow of information from a single input to an arbitrary collection of destinations within the circuit. Lemma \ref{lem: graphspdi}, a more general (assymetric) form of Lemma 2 in \cite{ES}, provides a certain form of tensorization of the mutual information between a single input variable and a collection of wires.

\begin{lemma}\label{lem: graphspdi}
    Let $G$ be a circuit composed of gates that are $(a,b)$-noisy as described in Theorem \ref{thm: gatespdi}, with inputs consisting of a collection of constants in $\{0,1\}$ and $X$, a single binary random variable.
    
    Let $\mathbf{W}$ be an arbitrary collection of wires in $G$, identified with the output (random values) they carry. It holds that
    \begin{align}
        I(X;\W)\le\sum_{X\xrightarrow[]{P}\W}c^{|P|},
    \end{align}
    where the sum is over all paths $P$ from $X$ to $\W$, $|P|$ the number of gates (excluding inputs) on the path $P$, and $c=c_{a,b}$ the information transmission bound from Corollary \ref{cor: conditioned}.

    For sake of clarity, a path from a vertex to a wire is taken to be a path to the originating vertex of the wire. Thus, paths passing through the same vertices are considered identical for the purpose of counting the number of paths. For the counting of path length, said originating vertex is included in the count.
\end{lemma}
\begin{proof}
Throughout this proof, a collection of possibly more than $1$ binary random variables is denoted in bold.

We first establish a numbering system to keep track of dependencies within the circuit. Assign every input vertex the number $0$. Since $G$ is a directed acyclic graph, it is possible to number the gate vertices distinctly, from $1$ to the number of gates in $G$, in such a way that every directed edge points away from a lower-numbered vertex and into a higher-numbered one. 

Such a numbering also induces a numbering of every wire by the number of its originating vertex. Note that the numbering of wires is not necessarily distinct as multiple wires may originate from the same vertex and thus receive the same number.

The proof proceeds by induction on the highest-numbered wire in $\W$. For the base case: if the highest-numbered wire is $0$, then every wire comes directly from an input. If there is a wire from $X$, then $I(X;\W)=1$, and the sum is no less than $1$ since said wire induces a path with no gates. If there is no wire from $X$, then $I(X;\W)=0$.

Suppose that, for $n\ge 0$, the desired inequality holds for every collection of wires whose highest number is no more than $n$. Let $\W$ be a collection of wires with highest number $n+1$. The analysis proceeds by partitioning $\W=\W'\sqcup W^+$, where $\W'$ is the collection of wires numbered up to $n$, and $W^+$ the collection of wires numbered $n+1$. Since all the wires in $W^+$ originate from the same gate and carry the same information, they may be regarded, information-theoretically, as a single wire. 

By the chain rule one has
\begin{align}
    I(X;\W)=&\ I(X;\W')+I(X;W^+\mid \W').
\end{align}
The first part can be bounded by the inductive hypothesis, and the latter part by examining what happens in gate $n+1$.

Let $Y$ be the pre-noise output of that gate (i.e., $W^+$ is the noised version of $Y$), and let $\Y$ be the input. We will relate $W^+$ first to $Y$ and then $\Y$, which consists of lower-numbered wires where the inductive hypothesis is applicable. Note that while $W^+$, $Y$, and $\Y$ are information-theoretically distinct, they are graph-theoretically all identified with vertex $n+1$; a schematic diagram of the gate numbered $n+1$ is provided in Figure \ref{fig: diagram}.


\begin{figure}[h!]
    \centering
   \tikzset{every picture/.style={line width=0.75pt}} 

\begin{tikzpicture}[x=0.75pt,y=0.75pt,yscale=-1.4,xscale=1.4]

\draw    (162.54,49.06) -- (246.09,87.24) ;
\draw    (162.54,97.84) -- (246.09,97.84) ;
\draw    (162.54,146.04) -- (246.09,106.52) ;
\draw   (246.09,77.98) -- (428.15,77.98) -- (428.15,116.54) -- (246.09,116.54) -- cycle ;
\draw    (428.15,107.09) -- (500,127.34) ;
\draw    (428.15,87.81) -- (500,68.54) ;
\draw   (155.86,45.4) .. controls (151.19,45.33) and (148.82,47.62) .. (148.75,52.29) -- (148.19,88.79) .. controls (148.08,95.46) and (145.7,98.75) .. (141.03,98.68) .. controls (145.7,98.75) and (147.98,102.12) .. (147.88,108.79)(147.92,105.79) -- (147.31,145.29) .. controls (147.24,149.96) and (149.53,152.33) .. (154.2,152.4) ;

\draw (272.76,88.53) node [anchor=north west][inner sep=0.75pt]   [align=left] {$\Y \xrightarrow[]{\text{compute}} Y \xrightarrow[]{\text{noise}} W^+$};
\draw (110.09,92.14) node [anchor=north west][inner sep=0.75pt]   [align=left] {$\Y$};
\draw (300.61,58.7) node [anchor=north west][inner sep=0.75pt]   [align=left] {gate $n+1$};
\draw (470.56,77.74) node [anchor=north west][inner sep=0.75pt]   [align=left] {$W^+$};
\draw (469.87,126.96) node [anchor=north west][inner sep=0.75pt]   [align=left] {$W^+$};
\draw (439.64,60.28) node [anchor=north west][inner sep=0.75pt]   [align=left] {wire $n+1$};
\draw (439.07,100.17) node [anchor=north west][inner sep=0.75pt]   [align=left] {wire $n+1$};

\end{tikzpicture}
    \caption{Schematic diagram of gate numbered $n+1$.}
    \label{fig: diagram}
\end{figure}
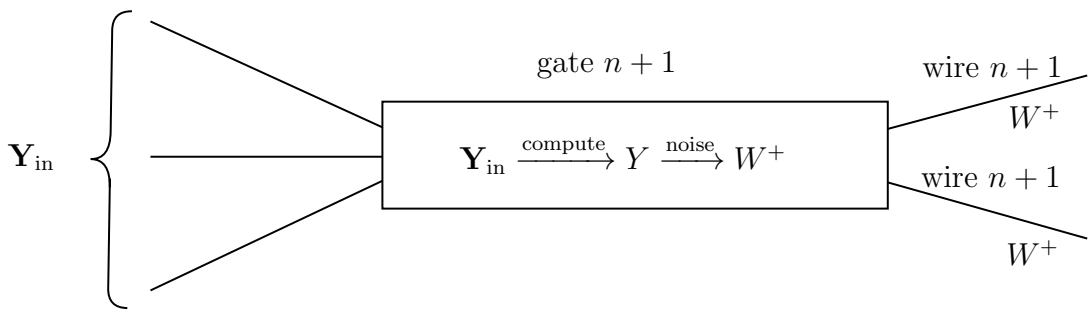

Since $X$ and every wire in $\W'$ is numbered less than $n+1$, they are conditionally independent of $W^+$ given $Y$, i.e., $(X,\W')\--Y\--W^+$ forms a Markov chain. By Corollary \ref{cor: conditioned},
\begin{align}
    I(X;W^+\mid \W')\le cI(X;Y\mid \W').
\end{align}
Furthermore, $(\W',X)\--\Y\--Y$ also forms a Markov chain, since $Y$ results from computing $\Y$, and every wire in $(\W',X)$ is numbered less than $n+1$. By the proof of Corollary \ref{cor: conditioned}, upon conditioning by $\W'$ being a fixed value the conditional random variables $X\--\Y\--Y$ still form a Markov chain, thus by the data processing inequality
\begin{align}
    I(X;Y\mid \W')\le&\ I(X;\Y\mid \W')\\
    =&\ I(X;\Y,\W')-I(X;\W').
\end{align}
We thus have
\begin{align}
    I(X;\W)\le&\ I(X;\W')+c(I(X;\Y,\W')-I(X;\W'))\\
    =&\ (1-c)I(X;\W')+cI(X;\Y,\W').
\end{align}
Since every wire in $\Y$ is input to gate $n+1$ and thus has a smaller number, the inductive hypothesis can be used to simply both terms, yielding
\begin{align}
    I(X;\W)\le&\ \left((1-c)\sum_{X\xrightarrow[]{P}\W'}+c\sum_{X\xrightarrow[]{P}\W',\Y}\right)c^{|P|}\\
    \le&\ \left((1-c)\sum_{X\xrightarrow[]{P}\W'}+c\sum_{X\xrightarrow[]{P}\W'}+c\sum_{X\xrightarrow[]{P}\Y}\right)c^{|P|}\\
    =&\ \left(\sum_{X\xrightarrow[]{P}\W'}+\sum_{X\xrightarrow[]{P}W^+}\right)c^{|P|}\\
    =&\ \sum_{X\xrightarrow[]{P}\W}c^{|P|},
\end{align}
noting that changing the destination from $\Y$ to $W^+$ increases the path length by $1$.
\end{proof}

Let $f$ be a function that depends on $n$ inputs. Suppose $C$ is some circuit, composed of $(a,b)$-noisy $k$-input gates, that computes $f$ with a probability of error no greater than $\delta$.

In the best case where the components are noiseless, the depth of a circuit grows necessarily with the number of inputs. Since a circuit of depth $d$ contains at most $k^d$ input vertices (achieved by a complete $k$-ary tree), a function that depends on $n$ variable inputs cannot be computed by any circuit of depth $d<\log_kn$. 

In the case of noisy gates, the depth bound given by \cite{ES} shows that the minimum depth of $C$ must be greater than the noiseless minimum $\log_kn$ by a multiplicative factor depending on the number of inputs per gate and the amount of noise.

The proof of the circuit depth bound relies on an observation about Lemma \ref{lem: graphspdi}: the bound it gives decreases exponentially along each path, and increases with each additional path. Thus, under certain conditions, one might expect a deeper circuit to provide enough additional paths to offset the decay due to depth. This heuristic argument can be made precise by the following bounds on minimum circuit depth and maximum component noise, which is a version of Theorem 2 in \cite{ES}.

\begin{theorem}\label{thm: depthbound}
    Let $f$ be a function that depends on $n$ inputs. Suppose $C$ is a circuit of depth $d$, composed of gates with at most $k$ inputs, with each gate being independently $(a,b)$-noisy. 
    
    Suppose $C$ computes $f$ $\delta$-reliably, i.e., with probability of error no more than $\delta$. Define $\Delta:=1-h(\delta)$, which can be seen as a measure of relevance between $C$'s output and the correct result. (Note that relevance is a better measure for our purpose than the probability of error--a circuit that errs with probability $1$ is easily corrected by inverting the result!)

    Let $c:=c_{a,b}$. If
    \begin{itemize}
        \item $ck> 1$, then $d\ge\frac{\log n\Delta}{\log ck}$.
        \item $ck\le 1$, then $\Delta\le\frac{1}{n}$.
    \end{itemize}
\end{theorem}
\begin{proof}
This proof uses the same essential ideas as the proof from \cite{ES}, but with a notable difference: a slight generalization of Lemma 3 from \cite{ES}, a purely graph-theoretic result, allows for the second part of the theorem to be proven much more readily, in parallel to the first part.

    Let $x_1,...,x_n$ be the $n$ inputs of $f$. Since $f$ depends on all of them, for every input $x_i$ there exists an assignment of the other inputs such that $f$, viewed as a function of $x_i$, is non-constant. Denote this restricted function by $f_i(x_i)$, which must equal $x_i$ or $1-x_i$, which are the only non-constant functions of $x_i$. Let $C_i$ be the circuit with the same restriction on the inputs other than $x_i$.

    Suppose we feed into $x_i$ a binary random variable with parameter $1/2$, denoted $X$. By Fano's inequality,
    \begin{align}
        H(X\mid C_i(X))\le&\ h(\P(C_i(X)\ne f_i(X)))+\P(C_i(X)\ne f_i(X))\log_2(2-1)\\
        \le&\ h(\delta),
    \end{align}
    thus
    \begin{align}
        I(X;C_i(X))=&\ H(X)-H(X\mid C_i(X))\\
        \ge&\ 1-h(\delta)=\Delta.
    \end{align}
    On the other hand, applying Lemma \ref{lem: graphspdi} with $G=C_i$ and $\W=C_i(X)$ yields
    \begin{align}
        I(X;C_i(X))\le\sum_{X\xrightarrow[]{P}C_i(X)}c^{|P|},
    \end{align}
    which combined with the previous inequality yields
    \begin{align}
        \Delta\le\sum_{X\xrightarrow[]{P}C_i(X)}c^{|P|}.
    \end{align}
    Summing over the $n$ input variables yields
    \begin{align}\label{eqn: sumbound}
        n\Delta\le\sum_{P\in C}c^{|P|},
    \end{align}
    where the sum is taken over all paths from an input vertex to the output.

    It is useful to find an upper bound for the sum; we prove a bound more general than Lemma 3 from \cite{ES}, which leads to a simpler treatment of the second part of the theorem. 
    
    \begin{lemma}\label{lem: treesum}
        Let $C$ be a circuit of depth $d$ composed of gates with $k$ or fewer inputs, and let $r\ge\frac{1}{k}$ be some constant. Then
        \begin{align}
            \sum_{P\in C}r^{|P|}\le (rk)^d.
        \end{align}
    \end{lemma}
\begin{proof}
    It suffices to show that among all circuits of depth $d$, the sum is maximized for the complete $k$-ary tree of depth $d$, which achieves $r^dk^d$.

    If $C$ is not a tree, it can be modified into a tree by duplicating every gate with multiple outputs so that each copy has only one output; this preserves the number and length of paths. The resulting tree can be turned into a complete tree of depth $d$ by completing non-leaf vertices with fewer than $k$ children, and by adding $k$ children to leaves of depth less than $d$. The former only introduces new paths, while the latter lengthens a path by $1$ and produces $k$ copies of it, thus scaling its contribution to the sum by a factor of $rk\ge1$.
\end{proof}

For the first part of the theorem: since $c>\frac{1}{k}$, take $r=c$ in Lemma \ref{lem: treesum}, which applied to inequality \ref{eqn: sumbound} yields
\begin{align}
    n\Delta\le\sum_{P\in C}c^{|P|}\le (ck)^d,
\end{align}
thus $d\ge \frac{\log n\Delta}{\log ck}$ as desired.
    
For the second part of the theorem: since $c\le\frac{1}{k}$, it holds that
\begin{align}
    n\Delta\le \sum_{P\in C}c^{|P|}\le \sum_{P\in C}\left(\frac{1}{k}\right)^{|P|}\le 1
\end{align}
upon taking $r=\frac{1}{k}$ in Lemma \ref{lem: treesum}. Thus $\Delta\le\frac{1}{n}$ as desired.
\end{proof}

There is an immediate takeaway from each part of this theorem. Firstly, note that compared to the minimum depth of $\frac{\log n}{\log k}$ for noiseless circuits, the depth lower bound obtained in Theorem \ref{thm: depthbound} is greater by a constant factor: we have
\begin{align}
    \left(\frac{\log n\Delta}{\log ck}\right)/\left(\frac{\log n}{\log k}\right)=&\ \frac{(\log n+\log\Delta)\log k}{(\log c+\log k)\log n}\\
    \approx&\ \frac{1}{1+\log c/\log k}
\end{align}
since $\log\Delta\approx0$ for high levels of reliability.

Secondly, one could use the critical threshold to bound the level of noise beyond which reliable computation is impossible. Observing that $c_{a,b}\le1-a-b$ for $a,b\le\frac{1}{2}$, reliable computation with $c_{a,b}> 1/k$ immediately requires that 
\begin{align}
    a+b<1-\frac{1}{k}.
\end{align} 
In the case that the noise is symmetric, i.e., $a=b$, one can obtain a more precise bound by noting $c_{a,a}=(1-2a)^2$, which forces 
\begin{align}
    a<\frac{1-1/\sqrt{k}}{2}.
\end{align}

\section{Noise control via majority gates} 

\subsection{Von Neumann's 3-Majority Gate}

Having introduced the noisy circuit model of computation, von Neumann asked whether it is possible for noisy circuits to simulate noiseless ones with reasonable accuracy. Observing that animal brains, despite being built from noisy components, are capable of reliably performing long chains of computation, von Neumann had guessed that such noise-controlling constructions must be possible. Indeed, von Neumann was able to demonstrate that a noisy circuit can simulate any noiseless circuit with reasonable accuracy at the cost of greater depth, provided that the error probability of each component stays below a critical threshold.

Von Neumann's construction relies on the 3-majority gate, which takes in $3$ binary inputs and outputs the value that appears twice or more. As a noise-controlling mechanism, it could be seen as a majority vote relying on the premise that it is more likely for a single result to be erroneous compared to two or more results out of three. Aside from being a noise-controlling mechanism, remarkably, the 3-majority gate, when combined with a NOT gate, also suffices for carrying out computation. Let $m(x,y,z)$ denote the majority of the input binary values $x,y,z$. Then one can express
\begin{align}
    \text{OR}(x,y)=&\ m(x,y,1),\\
    \text{AND}(x,y)=&\ m(x,y,0).
\end{align}
However, computation is not the focus of this essay, which instead aims to quantify the behavior of the 3-majority gate and establish necessary conditions for its noise-controlling purpose.

In von Neumann's analysis, noise at the gate level is assumed to be symmetric: bit flips from $0$ to $1$ and vice versa are equally likely. Thus, one can quantify the noise by a single parameter $p<1/2$, the probability that a bit flip occurs. One might imagine constructing $3$ independent and identical copies of said circuit, and feeding their outputs to a $3$-majority gate, hoping that the output of the gate errs with probability less than $p$. This is only possible if the gate itself errs with probability $\epsilon<p$, since the whole circuit cannot be more reliable than its final component.

The 3-majority gate construction can fail in two ways: the inputs to the gate are bad and the gate functions normally, or the inputs are good but the gate fails. Since the inputs err independently, the probability that $2$ or more err is 
\begin{align}
    \Theta(p):=p^2(1-p)\binom{3}{2}+p^3\binom{3}{3}=3p^2-2p^3.
\end{align}
The overall probability of failure is thus
\begin{align}
    f_{\epsilon}(p):=&\ \Theta(p)(1-\epsilon)+(1-\Theta(p))\epsilon\\
    =&\ \epsilon+(1-2\epsilon)(3p^2-2p^3).
\end{align}
Note that if the inputs each fail with probabilities less than $p$ or the gate fails with probability less than $\epsilon$, $f_{\epsilon}(p)$ is an upper bound for the probability of overall error. This holds due to the monotonicity of $\Theta(p)$ and that $1-\epsilon>\epsilon$.

\begin{figure}[h!]
    \centering
    \includegraphics[width=8cm]{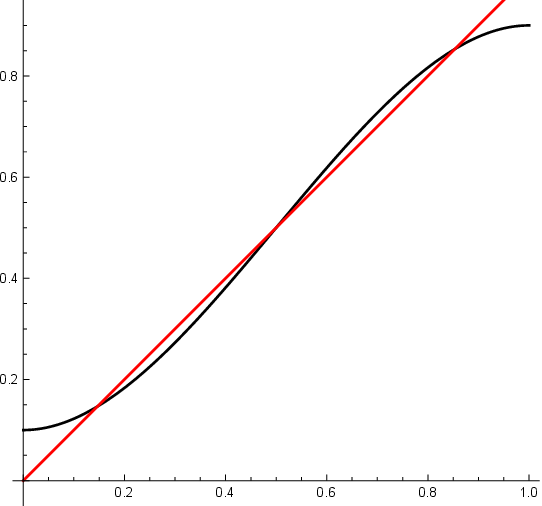}
    \caption{Plot of $f_{\epsilon}(p)$ over $p\in [0,1]$, with $\epsilon=0.1$. A plot of $p$ in red is provided for comparison. For $p\in(\approx 0.15, 0.5)$, $f_{\epsilon}(p)<p$ and noise reduction is possible.}
    \label{fig: errorRegion}
\end{figure}

The noise-controlling construction of triplicating the noisy circuit and feeding their outputs into a 3-majority gate works only if the probability of an erroneous gate output is less than the probability of an erroneous sample, i.e., $f_{\epsilon}(p)<p$. Figure \ref{fig: errorRegion} plots the two functions for a fixed $\epsilon=0.1$. In the interest of bounding the maximum allowable noise of individual components, von Neumann shows that
\begin{theorem}\label{thm: vN}
    For $\epsilon\ge\frac{1}{6}$, there does not exist any $p<\frac{1}{2}$ for which $f_{\epsilon}(p)<p$, i.e., that the triplication method of noise reduction works.
\end{theorem}
\begin{proof}
    Consider the difference in error probability 
    \begin{align}
        D_{\epsilon}(p):=&\ f_{\epsilon}(p)-p\\
        =&\ \epsilon+(1-2\epsilon)(3p^2-2p^3)-p.
    \end{align}

The function $D_{\epsilon}(p)$ is a cubic with a root at $p=1/2$. Since it is positive as $p\to-\infty$, having a noise reduction region where $D_{\epsilon}(p)$ is negative for some $p<1/2$ requires the existence of another real root. 

We can factor
\begin{align}
    \frac{D_{\epsilon}(p)}{(1/2-p)}=&\ (1-2\epsilon)p^2-(1-2\epsilon)p+\epsilon,
\end{align}
which has a real root only when the discriminant $(1-2\epsilon)^2-4(1-2\epsilon)\epsilon>0$, which is equivalent to $\epsilon<1/6$. (Note that using the discriminant simplifies von Neumann's analysis by avoiding the explicit computation of the roots.)
\end{proof}

\subsection{Higher-order majority gates}

Since majority gates rely on a ``popular vote" to decide on the correct input, it would be sensible to expect a larger number of inputs to yield a more accurate vote and thus an output with lower probability of being erroneous.

The $3$-majority gate introduced by von Neumann naturally generalizes to a $(2n+1)$-majority gate, which takes $2n+1$ binary inputs and outputs the value that occur at least $n+1$ times, i.e., more than half of the times. It is possible to build a similar noise-controlling mechanism where one starts off with some circuit with error probability $p$, construct $2n+1$ identical and independent copies of it, and feed their outputs to a $(2n+1)$-majority gate.

In this section, we extend von Neumann's noise threshold bound to majority gates of higher order; all results contained in this section are original to the author's knowledge. We begin by proving a direct generalization of Theorem \ref{thm: vN}.
\begin{theorem}\label{thm: generalvN}
     Let $2n+1$ binary results be generated in such a way that each each errs independently with probability $p<1/2$. Suppose that they are fed into a $(2n+1)$-majority gate that errs independently with probability $\epsilon$. The critical threshold for the gate noise, beyond which no margin of reduction is possible for any $p$, is 
     \begin{align}
         \epsilon<\frac{1-1/C(n)}{2}
     \end{align}
     where 
     \begin{align}
         C(n)=2^{-2n}(2n+1)\binom{2n}{n}.
     \end{align}
\end{theorem}
\begin{proof}
    In parallel to von Neumann's analysis, we obtain the probability that more than $n+1$ of the $2n+1$ inputs are erroneous, which is given by 
    \begin{align}
        \Theta_n(p):=\sum_{k=n+1}^{2n+1}\binom{2n+1}{k}p^k(1-p)^{2n+1-k};
    \end{align}
thus, the probability of the gate outputting an erroneous result is
\begin{align}
    \Theta_n(p)(1-\epsilon)+(1-\Theta_n(p))\epsilon=&\ \epsilon+(1-2\epsilon)\Theta_n(p).
\end{align}
Setting this lower than $p$, one has
\begin{align}\label{eqn: threshold}
    \Theta_n(p)<&\ \frac{1}{2}-\frac{1/2-p}{1-2\epsilon}.
\end{align}

\begin{figure}[h!]
    \centering
    \includegraphics[width=10cm]{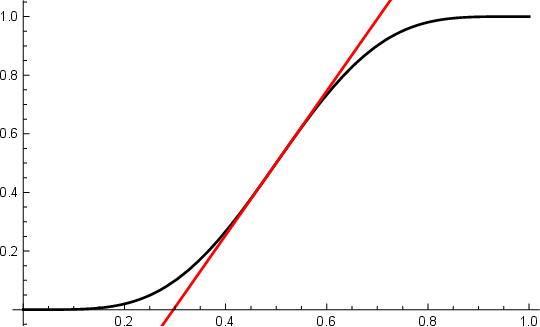}
    \caption{Plot of $\Theta_4(p)$ over $p\in[0,1]$, with the tangent at $1/2$ in red.}
    \label{fig: errorSlope}
\end{figure}

The right hand of \eqref{eqn: threshold}, as a function of $p$, defines a line with slope $\frac{1}{1-2\epsilon}$ that intersects with $\Theta_n(p)$ at $(1/2,1/2)$. Since $\Theta_n(p)$ is convex on $(0,1/2)$, a tangent line (or any line of greater slope) at $(1/2,1/2)$ does not meet $\Theta_n(p)$ in the interval; on the other hand, any line that has slope smaller than the tangent line but greater than $1$ is guaranteed an intersection within the interval. Figure \ref{fig: errorSlope} plots $\Theta_n(p)$ with its tangent line for $n=4$. 

The slope being greater than $1$ corresponds to $\epsilon>0$, so the threshold of interest is the derivative of $\Theta_n(p)$ at $1/2$. The derivative of each summand in $\Theta_n(p)$ is
\begin{align}
    &\ \binom{2n+1}{k}\left(kp^{k-1}(1-p)^{2n+1-k}-(2n+1-k)p^k(1-p)^{2n-k}\right)\bigg|_{p=1/2}\\
    =&\ 2^{-2n}\binom{2n+1}{k}(2k-2n-1),
\end{align}
thus
\begin{align}
    \Theta_n'\left(\frac{1}{2}\right)=&\ 2^{-2n}\sum_{k=n+1}^{2n+1}\binom{2n+1}{k}(2k-2n-1),
\end{align}
which can be shown to equal $C(n)$, see Lemma \ref{lem: combosum}.

Setting $\frac{1}{1-2\epsilon}$, the slope of the line in \eqref{eqn: threshold} less than $C(n)$ yields the desired inequality.
\end{proof}
Note that in the case of von Neumann's $3$-majority gate ($n=1$), Theorem \ref{thm: generalvN} produces the same threshold of $\epsilon<1/6$. Furthermore, one can show (see Lemma \ref{lem: asymptotic}) that for large $n$,
\begin{align}
    C(n)=\frac{2}{\sqrt{\pi}}\sqrt{n}+O\left(\frac{1}{\sqrt{n}}\right),
\end{align}
so the noise threshold we obtained is asymptotically
\begin{align}
    \epsilon\sim\frac{1-\sqrt{\frac{\pi}{4n}}}{2},
\end{align}
which is comparable to the threshold for $k$-input gates obtained by \cite{ES} by setting $k=2n+1$, which is asymptotically
\begin{align}
    \epsilon\sim\frac{1-\sqrt{\frac{1}{2n}}}{2}.
\end{align}
However, these asymptotics are not directly comparable because our result is specific to the majority gate and holds only on the level of one gate, while the results of \cite{ES} apply to arbitrary gates in an entire circuit.

Aside from error probability, another parameter of the majority gate is the number of inputs. Although \cite{ES} has shown that using gates with a large number of inputs reduces depth, they nevertheless introduce a great number of components, which is costly in both material and complexity. We give a sufficient bound for the number of inputs to achieve a specified level of noise reduction.

First is it useful to bound the probability $\Theta_n(p)$ that the inputs are bad, i.e., more than half of the $2n+1$ inputs are erroneous.


\begin{lemma}\label{lem: thetabound}
        For all $n\ge1$ it holds that
        \begin{align}
            \Theta_n(p)\le 2^{2n}(p(1-p))^n.
        \end{align}
    \end{lemma}
    \begin{proof}
       The number of erroneous results is a binomial random variable; a standard Chernoff bound (see Lemma \ref{lem: chernoff}) shows that for a binomial random variable $B_{N,p}$ with $t\ge N/2$, it holds that
    \begin{align}
        \P(B_{N,p}\ge t)\le N^N\left(\frac{p}{t}\right)^t\left(\frac{1-p}{N-t}\right)^{N-t}.
    \end{align}
    Plugging $N=2n+1$ and $t=n+1$, one has
    \begin{align}
        \P(B_{2n+1,p}\ge n+1)\le&\  (2n+1)^{2n+1}\left(\frac{p}{n+1}\right)^{n+1}\left(\frac{1-p}{n}\right)^{n}\\
        =&\ \frac{(2n+1)^{2n+1}}{n^n(n+1)^{n+1}}p^{n+1}(1-p)^n.
    \end{align}
    To simplify the coefficient, observe that $1+x\le e^x$ for all real $x$, from which we obtain
        \begin{align}
    \frac{(2n+1)^{2n+1}}{n^n(n+1)^{n+1}}=&\left(\frac{2n+1}{n}\right)^n\left(\frac{2n+1}{n+1}\right)^{n+1}\\
    =&\ 2^{2n+1}\left(1+\frac{1}{2n}\right)^n\left(1-\frac{1}{2(n+1)}\right)^{n+1}\\
    \le&\ 2^{2n+1}\left(e^{\frac{1}{2n}}\right)^n\left(e^{-\frac{1}{2(n+1)}}\right)^{n+1}\\
    =&\ 2^{2n+1}.
        \end{align}
    We thus have
    \begin{align}
        \P(B_{2n+1,p}\ge n+1)\le&\ 2^{2n+1}p^{n+1}(1-p)^n\\
        \le&\ 2^{2n}p^n(1-p)^n
    \end{align}
    as desired since $2p\le 1$.
    \end{proof}

    Now we can give a sufficient bound for the number of inputs.
    
\begin{theorem}
     Suppose that a majority gate can be built with error probability $\epsilon<1/2$, and that each input to the gate errs independently with probability $p<1/2$. 
    
    Suppose we want to make the output $\delta$-reliable, where $\epsilon<\delta<p$. (This is to ensure that the reliability level is achievable and nontrivial.) Then, it suffices to build the gate with $2n+1$ inputs as long as
    \begin{align}
        n\ge n(\epsilon,\delta,p):=\frac{\log\frac{1-2\epsilon}{\delta-\epsilon}}{\log \frac{1}{4p(1-p)}}.
    \end{align}
\end{theorem}
\begin{proof}
   By Lemma \ref{lem: thetabound}, the condition of $\delta$-reliability is guaranteed whenever
    \begin{align}
        \epsilon+(1-2\epsilon)\Theta_n(p)\le\epsilon+(1-2\epsilon)2^{2n}(p(1-p))^n\le\delta,
    \end{align}
    and extracting $n$ from the latter inequality yields the claimed result.
\end{proof}

\section{Concluding remarks}

In the analysis of noisy circuits, Evans and Schulman used SPDI's to bound the amount of information transmissible through each gate and provide a tensorization of mutual information between any input and any collection of wires, culminating in a bound on minimum circuit depth and maximum component noise, beyond which reliable computation is impossible. They found that in order to simulate with reliability the computation of a noiseless circuit, any noisy circuit must be at least a constant factor deeper, with the factor varying positively with the amount of noise per gate, quantified as the information transmission factor $c_{a,b}$, and varying negatively with the number $k$ of inputs per gate. 

It seems at first counterintuitive that an increase in depth is necessary for reliability, since information transmission along each path decays exponentially with depth; however, we saw that the decay due to path length can be offset by increasing the number of paths, i.e., the number of inputs to each gate. The point of criticality occurs when the product $c_{a,b}k$ equals $1$; if the noise cannot be made smaller or the number of inputs not greater, catastrophic failure is guaranteed.

Prior to Evans and Schulman, von Neumann proposed the 3-majority gate as a noise-controlling mechanism and demonstrated that it is possible to simulate any noiseless circuit by a noisy one that is a constant factor deeper. Thus, the results of Evans and Schulman can be viewed as a confirmation that von Neumann's construction is, in a weak sense, optimal. 

We saw that the $3$-majority gate fails to perform noise reduction if the error probability per component is not less than $1/6$. However, this constant is specific to the $3$-majority gate; by extending von Neumann's analysis to majority gates of any odd order, we saw that any error rate less than $1/2$ can be controlled to arbitrarily approach the error rate of the gate, provided that the gate has a sufficiently large number of inputs. We also found, as a function of input size, the noise threshold beyond which no noise reduction is possible, which, remarkably, has the same growth rate as the much more general threshold derived by Evans and Schulman.

The work of \cite{ES} and \cite{vN} provokes many interesting questions that are not yet resolved to a satisfactory extent.

\begin{itemize}
    \item The work of \cite{ES}, in the form of Theorem \ref{thm: depthbound}, gave a lower bound on circuit depth $d$, and the critical threshold itself can be seen as a lower bound on number of inputs per gate $k$. Both quantities are related to latency, i.e., the time requires to run the circuit: deeper circuits require longer sequences of computation, while ``wider" circuits consisting of gates with more inputs require more time to complete the computation at each gate.
    
    As we have seen, one can in a sense ``cheat" one complexity metric by reducing it and increasing the other instead: increasing $k$ results in a smaller lower bound for $d$. It would be illuminating to find a lower bound for a more reliable notion of complexity, such as the total number of gates or wires, that is not in tradeoff with another complexity metric.
    
    \item For a $k$-input gate, Evans and Schulman give an noise threshold of $\epsilon<\frac{1-1/\sqrt{k}}{2}$ above which reliable computation is impossible. Conversely, one may ask whether this bound is optimal, i.e., for every permissible noise level there exists a way to design a noisy circuit that can reliably simulate any noiseless circuit, or, alternatively, if there exists a lower noise threshold in practice. Note that von Neumann's $3$-majority gate fails at $\epsilon=1/6\approx 0.17$, which is lower than Evans and Schulman's $\frac{1-1/\sqrt{3}}{2}\approx 0.21$; similarly, the asymptotic threshold of $\frac{1-\sqrt{\pi/(4n)}}{2}$ for $(2n+1)$-majority gates is less than the asymptotic $\frac{1-\sqrt{1/(2n)}}{2}$ implied by Evans and Schulman.
    \item The results of \cite{ES} are highly suggestive of an analogy comparing information in a circuit to electrical voltage: both decay (albeit at different rates) over long distances: the former to noise, the latter to resistance, and both transmit more efficiently as more paths become available. Furthermore, the tensorization of mutual information over paths provided in Lemma \ref{lem: graphspdi} vaguely resembles the additivity of conductance for resistors in parallel. It would be interesting to see whether a precise analogy can be made, with electricity or another physical quantity.
\end{itemize}

\appendix

\section{Auxiliary computations}

This Appendix collects computations that support results from the main sections but are long and not immediately relevant to the main ideas of this essay.

\begin{lemma}\label{lem: limitingdist}
    Let $Y$ be a binary random variable, and let $Y_n$ be the result of passing $Y$ through an $(a,b)$-noisy channel $n$ times. As $n\to\infty$, it holds that
    \begin{align}
    \begin{pmatrix}
        \P(Y_n=0)\\
        \P(Y_n=1)
    \end{pmatrix}
        \to
        \begin{pmatrix}
            a/(a+b)\\
            b/(a+b)
        \end{pmatrix}.
    \end{align}
\end{lemma}
\begin{proof}
    Observe that one can write
    \begin{align}
        \begin{pmatrix}
            \P(Y_{n+1}=0)\\
            \P(Y_{n+1}=1)
        \end{pmatrix}
        =&\begin{pmatrix}
            \P(Y_{n+1}=0\mid Y_{n}=0) & \P(Y_{n+1}=0\mid Y_{n}=1)\\
            \P(Y_{n+1}=1\mid Y_{n}=0) & \P(Y_{n+1}=1\mid Y_{n}=1)
        \end{pmatrix}
        \begin{pmatrix}
            \P(Y_{n}=0)\\
            \P(Y_{n}=1)
        \end{pmatrix}\\
        =&\begin{pmatrix}
            1-a & b\\
            a & 1-b
        \end{pmatrix}
        \begin{pmatrix}
            \P(Y_{n}=0)\\
            \P(Y_{n}=1)
        \end{pmatrix}.
    \end{align}
The transition matrix can be diagonalized as
\begin{align}
    \begin{pmatrix}
            1-a & b\\
            a & 1-b
        \end{pmatrix}
        =\begin{pmatrix}
            a/b & -1\\
            1 & 1
        \end{pmatrix}
        \begin{pmatrix}
            1 & 0\\
            0 & 1-a-b
        \end{pmatrix}
        \begin{pmatrix}
            b/(a+b) & b/(a+b)\\
            -b/(a+b) & a/(a+b)
        \end{pmatrix},
\end{align}
whose powers approach
\begin{align}
    \begin{pmatrix}
            a/b & -1\\
            1 & 1
        \end{pmatrix}
        \begin{pmatrix}
            1 & 0\\
            0 & 0
        \end{pmatrix}
        \begin{pmatrix}
            b/(a+b) & b/(a+b)\\
            -b/(a+b) & a/(a+b)
        \end{pmatrix}
        =\begin{pmatrix}
            a/(a+b) & a/(a+b)\\
            b/(a+b) & b/(a+b)
        \end{pmatrix}.
\end{align}
Thus we have
\begin{align}
     \begin{pmatrix}
            \P(Y_{n}=0)\\
            \P(Y_{n}=1)
        \end{pmatrix}\to&
        \begin{pmatrix}
            a/(a+b) & a/(a+b)\\
            b/(a+b) & b/(a+b)
        \end{pmatrix}
        \begin{pmatrix}
            \P(Y=0)\\
            \P(Y=1)
        \end{pmatrix}\\
        =&\begin{pmatrix}
            a/(a+b)\\
            b/(a+b)
        \end{pmatrix}.
\end{align}
\end{proof}
\begin{lemma}\label{lem: combosum}
    For all $n$ it holds that
    \begin{align}
        \sum_{k=n+1}^{2n+1}\binom{2n+1}{k}(2k-2n-1)=(2n+1)\binom{2n}{n}.
    \end{align}
\end{lemma}
\begin{proof}
    Since
    \begin{align}
        \sum_{k=n+1}^{2n+1}\binom{2n+1}{k}=2^{2n}
    \end{align}
    by symmetry in the lower argument of binomial coefficients, the sum boils down to evaluating
    \begin{align}
        S_n:=\sum_{k=n+1}^{2n+1}\binom{2n+1}{k}k,
    \end{align}
    noting that the desired sum can be expressed as $2S_n-(2n+1)2^{2n}$.

    Each summand can be rewritten as
    \begin{align}
        \binom{2n+1}{k}k=&\ \frac{(2n+1)!}{k!(2n-k+1)!}k\\
        =&\ (2n+1)\frac{(2n)!}{(k-1)!(2n-(k-1))!}\\
        =&\ (2n+1)\binom{2n}{k-1},
    \end{align}
    and re-indexing the sum via $m=k-1$ gives
    \begin{align}
        S_n=&\ (2n+1)\sum_{m=n}^{2n}\binom{2n}{m}.
    \end{align}
    By symmetry of the lower argument, again, it follows that
    \begin{align}
        \sum_{m=n}^{2n}\binom{2n}{m}=2^{2n-1}+\frac{1}{2}\binom{2n}{n},
    \end{align}
    which allows the desired sum to be written as
    \begin{align}
        2\left((2n+1)\left(2^{2n-1}+\frac{1}{2}\binom{2n}{n}\right)\right)-(2n+1)2^{2n}=(2n+1)\binom{2n}{n}
    \end{align}
    as required.
\end{proof}

\begin{lemma}\label{lem: asymptotic}
    Let 
    \begin{align}
        C(n):=2^{-2n}(2n+1)\binom{2n}{n}.
    \end{align}
    For large $n$ it holds that
    \begin{align}
       C(n)=\frac{2}{\sqrt{\pi}}\sqrt{n}+O\left(\frac{1}{\sqrt{n}}\right).
    \end{align}
\end{lemma}
\begin{proof}
    We begin with an estimate of the central binomial coefficient
    \begin{align}
        \binom{2n}{n}=\frac{2^{2n}}{\sqrt{\pi n}}\left(1+O\left(\frac{1}{n}\right)\right),
    \end{align}
    which is a rearrangement of Equation (20) from \cite{YL}, page 35. (One can alternatively derive it from Stirling's formula with a $O(1/n)$ error term.) Substituting this into $C(n)$ gives
    \begin{align}
        C(n)=\frac{2n+1}{\sqrt{\pi n}}\left(1+O\left(\frac{1}{n}\right)\right)=\frac{2}{\sqrt{\pi}}\sqrt{n}+O\left(\frac{1}{\sqrt{n}}\right)
    \end{align}
    as required.
\end{proof}

\begin{lemma}\label{lem: chernoff}
For a binomial random variable with parameters $N, p<1/2$, it holds that for $t\ge N/2$,
    \begin{align}
        \P(B_{N,p}\ge t)\le N^N\left(\frac{p}{t}\right)^t\left(\frac{1-p}{N-t}\right)^{N-t}.
    \end{align}
\end{lemma}
\begin{proof}
The binomial distribution can be decomposed as a sum of independent copies of a Bernoulli distribution $B_p$, which has log-MGF
    \begin{align}
        \psi_{B_p}(\lambda)=\log(pe^{\lambda}+(1-p)).
    \end{align}
    One thus has, for a binomial random variable $B_{N,p}$, that
    \begin{align}
        \psi(\lambda):=\psi_{B_{N,p}}(\lambda)=N\log(pe^{\lambda}+(1-p)).
    \end{align}

    The Chernoff-Cram\'er method gives an upper bound
    \begin{align}
        \P(B_{N,p}\ge t)\le e^{\psi(\lambda^*)-\lambda^* t}
    \end{align}
    where 
    \begin{align}
        \lambda^*=\argmax_{\lambda>0} \lambda t-\psi(\lambda).
    \end{align}
    Differentiating gives a candidate $\lambda^*$ satisfying
    \begin{align}
        t-N\frac{pe^{\lambda^*}}{pe^{\lambda^*}+(1-p)}=&\ 0,\\
        e^{\lambda^*}=&\ \frac{(1-p)t}{p(N-t)}.
    \end{align}
    Since $1-p>p$ and $t\ge N-t$, $\lambda^*$ is indeed positive, which gives the bound
    \begin{align}
        \P(B_{N,p}\ge t)\le&\ \left(pe^{\lambda^*}+(1-p)\right)^N(e^{\lambda^*})^{-t}\\
        =&\ \left(\frac{N(1-p)}{N-t}\right)^N\left(\frac{(1-p)t}{p(N-t)}\right)^{-t}\\
        =&\ N^N\left(\frac{p}{t}\right)^t\left(\frac{1-p}{N-t}\right)^{N-t}
    \end{align}
    as required.
\end{proof}

\end{document}